%% file: RR-8793.tex
\documentclass[twoside]{article}
\usepackage[a4paper]{geometry}

\usepackage[T1]{fontenc} 
\usepackage{RR}
\usepackage[bookmarks,colorlinks=true,
  citecolor=black,
  linkcolor=black,
  urlcolor=black]{hyperref}

\usepackage{amsmath, amssymb, amsfonts,amsthm}
\usepackage{booktabs}
\usepackage{color}
\usepackage{pifont}
\usepackage{multirow}
\usepackage{xspace}
\usepackage{url}
\usepackage{tabularx}
\usepackage{mathtools}
\newcommand{\bP}{{\bf P}}

\newcommand{\bD}{{\bf D}}
\newcommand{\bI}{{\bf I}}
\newcommand{\bZ}{{\bf Z}}

\newcommand{\cD}{\mathcal{D}}

\newcommand{\cS}{\mathcal{S}}
\newcommand{\csn}{{\mathcal{S}_n}}

\newtheorem{theorem}{Theorem}
\newtheorem{lemma}{Lemma}

\theoremstyle{definition}

\newtheorem{remark}{Remark}

\newcommand{\eat}[1]{}

\newcommand{\bone}{\boldsymbol{1}}
\DeclareMathOperator*{\argmax}{arg\,max}

\RRdate{October 2015}

\RRauthor{
Konstantin Avrachenkov\thanks{Inria Sophia Antipolis, France. {Email: \tt k.avrachenkov@sophia.inria.fr}}%
  \and
Bruno Ribeiro\thanks{Purdue University, IN, USA {Email: \tt ribeiro@cs.purdue.edu}}  
\and \\
Jithin K.\ Sreedharan\thanks{Inria Sophia Antipolis, France. {Email: \tt jithin.sreedharan@inria.fr}}\thanks{Corresponding author}
}

\authorhead{K.\ Avrachenkov, B.\ Ribeiro \& J.\ K.\ Sreedharan}
\RRtitle{Inf\'erence bay\'esienne de statistiques des r\'eseaux sociaux par les marches al\'eatoires avec complexit\'e l\'eg\`ere}
\RRetitle{Bayesian Inference of Online Social Network Statistics via Lightweight Random Walk Crawls}
\titlehead{Bayesian Inference of Online Social Network Statistics via Lightweight Random Walk Crawls}
\RRresume{Les r\'eseaux sociaux en ligne contiennent grande quantit\'e d'informations \`a propos de la soci\'et\'e qui peut \^etre significativement explor\'e. Une des techniques, parmi les plus faisable de r\'ecup\'erer des informations, est d'explorer le r\'eseau \`a l'aide de Application Programming Interface (API). Dans ce travail, nous nous concentrons sur l'inf\'erence statistique fiable avec un taux limit\'e des appels vers l'API. Bas\'e sur les propri\'et\'es r\'eg\'en\'eratrices des marches al\'eatoires, nous proposons un estimateur sans bias pour la somme agr\'eg\'ee des fonctions sur des ar\^etes. Nous montrons la connexion entre la variance de l'estimateur et l'\'ecart spectral. Afin de faciliter l'inf\'erence bay\'esienne sur la vraie valeur de l'estimateur, nous d\'erivons la distribution post\'erieure asymptotique de l'estimation. Finalement, les id\'ees propos\'ees sont valid\'es par des exp\'eriences num\'eriques sur les probl\`emes d'inf\'erence dans les r\'eseaux sociaux r\'eels}
\RRabstract{
Online social networks (OSN) contain extensive amount of information about the underlying society that is yet to be explored. One of the most feasible technique to fetch information from OSN, crawling through Application Programming Interface (API) requests, poses serious concerns over the the guarantees of the estimates. In this work, we focus on making reliable statistical inference with limited API crawls. Based on regenerative properties of the random walks, we propose an unbiased estimator for the aggregated sum of functions over edges and proved the connection between variance of the estimator and spectral gap. In order to facilitate Bayesian inference on the true value of the estimator, we derive the approximate posterior distribution of the estimate. Later the proposed ideas are validated with numerical experiments on inference problems in real-world networks. 
}
\RRmotcle{Inf\'erence bay\'esienne, marche al\'eatoire sur les graphes, analyse de r\'eseau social.}
\RRkeyword{Bayesian inference, Random walk on graphs, Social network analysis, Graph sampling}
\RRprojet{Maestro}
\RCSophia 

\begin{document}
\RRNo{8793}
\makeRR   

\section{Introduction}
What is the fraction of male-female connections against that of female-female connections in a given Online Social Network (OSN)? Is the OSN  assortative or disassortative?
Edge, triangle, and node statistics of OSNs find applications in computational social science (see e.g.~\cite{ottoni2013ladies}), epidemiology~\cite{pastor2001epidemic}, and computer science~\cite{benevenuto2009characterizing,gjoka20132,putzke2014cross}.
Computing these statistics is a key capability in large-scale social network analysis and machine learning applications.
But because data collection in the wild is often limited to partial OSN crawls through Application Programming Interface (API) requests,
observational studies of OSNs -- for research purposes or market analysis -- depend in great part on our ability to compute network statistics with incomplete data.
Case in point, most datasets available to researchers in widely popular public repositories are partial OSN crawls\footnote{The majority of the datasets in the public repositories SNAP~\cite{Snapdata} and KONECT~\cite{kunegis2013konect} are partial website crawls, not complete datasets or uniform samples.}.
Unfortunately, these incomplete datasets have unknown biases and no statistical guarantees regarding the accuracy of their statistics.
To date, the best methods for crawling networks~(\cite{avrachenkov2010improving,gjoka2011practical,Ribeiro2010}) show good real-world performance but only provide statistical guarantees asymptotically (i.e., when the entire OSN network is collected).

This work addresses the fundamental problem of obtaining unbiased and reliable node, edge, and triangle statistics of OSNs via partial crawling.
{\em To the best of our knowledge our method is the first to provide a practical solution to the problem of computing OSN statistics with strong theoretical guarantees from a partial network crawl.}
More specifically, we (a) provide a provable finite-sample unbiased estimate of network statistics (and their spectral-gap derived variance) and (b) provide the asymptotic posterior of our estimates that performs remarkably well all tested real-world scenarios.

More precisely, let $G=(V,E)$ be an undirected labeled network -- not necessarily connected -- where $V$ is the set of vertices and $E \subseteq V \times V$ is the set of edges.
Both edges and nodes can have labels.
Network $G$ is unknown to us except for $n > 0$  {\em arbitrary}  initial seed nodes in $I_n \subseteq V$.
Nodes in $I_n$ must span all the different connected components of $G$.
From the seed nodes we crawl the network starting from $I_n$ and obtain a set of crawled edges $\cD_m(I_n)$, where $m > 0$ is a parameter that regulates the number of website API requests.
With the crawled edges $\cD_m(I_n)$ we seek an unbiased
estimate of
\begin{equation}\label{e:mu}
\mu(G) = \sum_{(u,v) \in E} f(u,v) \, .
\end{equation}
Note that functions of the form eq.~\eqref{e:mu} are general enough to compute node statistics
\[
\mu_\text{node}(G) = \sum_{(u,v) \in E} g(v)/d_v \, ,
\]
where $d_u$ is the degree of node $u \in V$,
and statistics of triangles such as the local clustering coefficient of $G$ first provided by~\cite{Ribeiro2010}
\[
\mu_\bigtriangleup(G) = \frac{1}{|V|} \sum_{(u,v) \in E}  \frac{{\bf 1}(d_v > 2)}{d_v}  \frac{\sum_{a \in N_v}\sum_{b \in N_v, b \neq a} {\bf 1}((v,a) \in E \cap (v,b) \in E \cap (a,b) \in E) }{\binom{d_v}{2}}  \, ,
\]
where the expression inside the sum is zero when $d_v < 2$ and $N_v$ are the neighbors of $v \in V$ in $G$.
Our task is to find estimates of general functions of the form $\mu(G)$ in eq.~\eqref{e:mu}.

\subsection*{Contributions}
In our work we provide a partial crawling strategy using random walk tours whose posterior
\[
P[\mu(G) | \cD_m(I_n)] \,
\]
is shown to have an unbiased maximum a posteriori estimate (MAP) $\hat{\mu}_\text{MAP}(\cD_m(I_n))$ regardless of the number of nodes in the seed set $n > 0$ and regardless of the value of $m > 0$, i.e., $E[\hat{\mu}_\text{MAP}(\cD_m(I_n))] = \mu(G),$ $\forall n,m > 0$.  Note that we guarantee that our MAP estimate is unbiased in the finite-sample regime unlike previous asymptotic methods~\cite{avrachenkov2010improving,gjoka2011practical,lee2012beyond,Ribeiro2010,ribeiro2012sampling}. Moreover, we provide the posterior $P[\mu(G) | \cD_m(I_n)]$ for the large $m$ regime and prove its convergence in distribution showing its convergence rate.
In our experiments we note that the posterior is remarkably accurate using a variety of networks large and small.
We also provide upper and lower bounds for $P[\mu(G) | \cD_m(I_n)]$.

\subsection*{Related Work}
The works of \cite{Massoulie2006} and \cite{Cooper2013} are the ones closest to ours. \cite{Massoulie2006} estimates the size of a network based on the return times of random walk tours. \cite{Cooper2013} estimates number of triangles, network size, and subgraph counts from weighted random walk tours using results of \cite{Aldous_2014}. The previous works on non-asymptotic inference of network statistics from incomplete network crawls~\cite{Goel2009,Koskinen2010,Koskinen2013,Handcock2010,Heckathorn:1997aa, ligo2014controlled, thompson2006targeted} need to fit the partial observed data to a probabilistic graph model such as ERGMs (exponential family of random graphs models). Our work advances the state-of-the-art in estimating network statistics from partial crawls because: (a) we estimate statistics of arbitrary edge functions without assumptions about the graph model or the underlying graph; (b) we do not need to bias the random walk with weights; this is particularly useful when estimating multiple statistics reusing the same observations; (c) we derive upper and lower bounds on the variance of estimator, which both show the connection with the spectral gap; and, finally, (d) we compute a posterior over our estimates to give practitioners a way to access the confidence in the estimates without relying on unobtainable quantities like the spectral gap and without assuming a probabilistic graph model.

The remainder of the paper is organized as follows. In Section~\ref{s:theory} we introduce our main theorems and supporting lemmas and proofs.
In Section~\ref{s:illustrative_examples} we introduce artificial illustrative examples to aid understanding our method.
In Section~\ref{s:results} we introduce our results using simulations over real-world networks.
Finally, in Section~\ref{s:conclusions} we present our conclusions.

\section{Network Estimation from Partial Crawls}
\label{s:theory}

\input{theory_tours}

\input{theory_estimator}

\input{illustrative_examples}

\input{real_examples}

\input{conclusions}

\section*{Acknowledgements}

The authors acknowledge the partial support by ADR "Network Science" of Inria - Bell Labs and by the joint Brazilian-French research team Thanes. This work was in part supported by NSF
grant CNS-1065133 and ARL Cooperative Agreement W911NF-09-2-0053. The views and conclusions contained in this document
are those of the author and should not be interpreted as representing
the official policies, either expressed or implied of the
NSF, ARL, or the U.S. Government. The U.S. Government is
authorized to reproduce and distribute reprints for Government
purposes notwithstanding any copyright notation hereon.

\bibliographystyle{plain}
\bibliography{hyptestgraph}

\section{Appendix}
\input{appendix_tours}
\end{document}

%% file: theory_tours.tex

In this section we present our main results.
The outline of this section is as follows.
Section~\ref{s:pre} introduces key concepts and defines the notation used throughout this manuscript.
Section~\ref{s:main} introduces our main results in the form of two theorems:
Theorem~\ref{th:unbias} presents an unbiased estimator of any function over edges of an undirected graph using random walk tours.
Our random walk tours are shorter than the ``regular random walk tours'' because the ``node'' that they start from is an amalgamation of a multitude of nodes in the graph. Here, we briefly explains the approximate posterior of the estimator in Theorem~\ref{th:unbias}.
Section~\ref{s:proof1} proves Theorem~\ref{th:unbias}, introducing important upper and lower bonds of the
estimator variance in Section~\ref{s:var} and showing the effect of the spectral gap.
Finally, Section~\ref{s:BayesianInference} derives the approximate Bayesian posterior \eqref{e:post} also using the bounds obtained in Section~\ref{s:var}.

\subsection{Preliminaries}\label{s:pre}
Let $G = (V,E)$ be an unknown undirected graph.
Our goal is to find an unbiased estimate of $\mu(G)$ in eq.~\eqref{e:mu}
and its posterior by crawling a small fraction of $G$.
We are given a set of $n > 0$ initial {\em arbitrary} nodes denoted $I_n \subset V$.
If $G$ has disconnected components $I_n$ must span all the different connected components of $G$.

Our network crawler is a classical random walk over the following augmented multigraph $G'=(V',E')$.
A multigraph is a graph that can have multiple edges between two nodes.
In $G'$ we aggregate all nodes of $I_n$ into a single node, denoted hereafter $\cS_n$, the {\em super-node}.
Thus, $V' = \{V \backslash I_n\} \cup \{\cS_n\}$.
The edges of $G'$ are $E' = E \backslash \left\{E \cap \{I_n \times V\} \right\} \cup \{(\cS_n,v) : \forall (u,v) \in E, \text{ s.t.\ } u\in I_n\text{ and }v \in V \backslash I_n\}$, i.e., $E'$ contains all the edges in $E$ including the edges from the nodes in $I_n$ to other nodes, and $I_n$ is merged into the super-node $S_n$.
Note that $G'$ is necessarily connected as $I_n$ spans all the connected components of $G$.

A random walk on $G'$ has transition probability from node $u$ to an adjacent node $v$, $p_{uv}:=\bP_{u,v}$ with $\alpha_{u,v}/d_{u}$, where $d_u$ is the degree of $u$ and $\alpha_{u,v}$ is the number of edges between $u \in V'$ and $v \in V'$.
We note that the theory presented in the paper can be extended to more sophisticated random walks as well.
Let $\pi_i$ be the stationary distribution at node $i$ in the random walk on $G'$. 

A random walk \textit{tour} is defined as the sequence of nodes $X_1^{(k)},\ldots,X_{\xi_k}^{(k)}$ visited by the random walk during successive $k$-th and $k+1$-st visits to the super-node $\cS_n$. Here $\{\xi_k\}_{k\geq 1}$ denote the successive return times to $\csn$. Tours have a key property: from the renewal theorem tours are independent since the returning times act as renewal epochs.
Moreover, let $Y_1,Y_2,\ldots,Y_n$ be a random walk on $G'$ in steady state.

Note that the random walk on $G'$ is equivalent to a random walk on $G$ where all the nodes in $I_n$ are treated as {\bf one single node}.

The function $f$ is redefined on $G'$ as follows: for $(u,v) \in E'$, $f(u,v))$ remains same when $u \notin \csn$ and $v \notin \csn$. But when $u \in \csn$ or $v \in \csn$, $f(u,v)$ is redefined as zero.

\subsubsection*{Super-node Motivation}
The introduction of super-node is primary motivated by the following three reasons:
\begin{itemize}
\item { \em Tackling disconnected or low-conductance graphs:} When the graph is not strongly connected or has many connected components, forming a super-node with representatives from each of the components make the modified graph connected and suitable for applying random walk theory. Even when the graph is connected, it might not be well-knit, i.e., it has low conductance. Since the conductance is closely related to mixing time of Markov chains, such graph will prolong the mixing of random walks. But with proper choice of super-node, we can reduce the mixing time and, as we show, improve the estimation accuracy. This idea is illustrated with a Dumbell graph example in Section \ref{s:illustrative_examples}.
\item {\em Faster Estimate with Shorter Tours:} The expected value of the $k$-th tour length $E[\xi_k]=1/\pi_{\cS_n}$ is inversely proportional to the degree of the super-node $d_{\cS_n}$. Hence, by forming a massive-degree super-node we can significantly shorten the average tour length.
\end{itemize}

\subsection{Main Results} \label{s:main}
In what follows we present our main results. Theorem~\ref{th:unbias} proposes an unbiased estimator of edge characteristics $\mu(G)$ via random walk tours. Then we present the approximate posterior distribution of the unbiased estimator presented in Theorem~\ref{th:unbias}.
\begin{theorem}
\label{th:unbias}
Let $G$ be an unknown undirected graph where $n > 0$ initial {\em arbitrary} set of nodes is known $I_n \subseteq V$ which span all the different connected components of $G$.
Consider a random walk on the augmented multigraph $G'$ described in Section~\ref{s:pre} starting at super-node $\cS_n$.
Let $(X_t^{(k)})_{t=1}^{\xi_k}$ be the $k$-th random walk tour until the walk first returns to $\cS_n$ and let
$\cD_m(\cS_n)$ denote the collection of all nodes in $m \geq 1$ such tours, $\cD_m(\cS_n) = \left( (X_t^{(k)})_{t=1}^{\xi_k}\right)_{k=1}^m$.
Then,
\begin{equation}\label{e:hatmu}
\hat{\mu}(\cD_m(\cS_n)) = \overbrace{\frac{d_{\cS_n}}{2m} \sum_{k=1}^m \sum_{t=2}^{\xi_k} f(X^{(k)}_{t-1},X^{(k)}_t)}^\text{Estimate from crawls} \quad + \overbrace{\sum_{\mathclap{\substack{(u,v) \in E \text{ s.t.\ } \\ u\in I_n \text{ or } v \in I_n }}}^{\vphantom{\xi_k}} f(u,v)}^\text{Edges between initial nodes at original $G$}
\end{equation}
is an unbiased estimate of $\mu(G)$, i.e., $E[\hat{\mu}(\cD_m(\cS_n))] = \mu(G)$. Moreover the estimator is strongly consistent, i.e., $\hat{\mu}(\cD_m(\cS_n)) \to \mu(G)$ a.s.\ for $m\to \infty$.
\end{theorem}
Theorem~\ref{th:unbias} provides an unbiased estimate of network statistics from random walk tours.
The length of tour $k$ is short if it starts at a massive super-node as the expected tour length is inversely proportional to the degree of the super-node, $E[\xi_k] \propto 1/d_{\cS_n}$.
This provides a practical way to compute unbiased estimates of node, edge, and triangle statistics using $\hat{\mu}(\cD_m(\cS_n))$ (eq.~\eqref{e:hatmu}) while observing only a small fraction of the original graph.
Because random walk tours can have arbitrary lengths, we show in Lemma \ref{th:variance_tour_estr}, Section~\ref{s:BayesianInference}, that there are upper and lower bounds on the variance of $\hat{\mu}(\cD_m(\cS_n))$. For a bounded function $f$, the upper bounds are shown to be always finite.

In what follows we show the approximate posterior of the estimator in Theorem~\ref{th:unbias}.
In Section~\ref{s:results} we shall see that the approximate posterior matches very well the empirical posterior using simulations over real-world  networks while crawling $<10\%$ of the nodes in the network.

Let ${\mu}(G')$ be the true value $\mu(G)$ outside the subgraph formed by the nodes that were merged into the super-node.

\subsubsection*{Bayesian approximation of the posterior of $\mu(G)$}
Let
\[
\hat{F}_h = \frac{d_{\cS_n}}{2 \lfloor\sqrt{m}\rfloor} \sum_{k=((h-1)\lfloor \sqrt{m}\rfloor + 1)}^{h \lfloor\sqrt{m}\rfloor} \sum_{t=2}^{\xi_h}  f(X^{(k)}_{t-1},X^{(k)}_{t})+\sum_{\mathclap{\substack{(u,v) \in E \text{ s.t.\ } \\ u\in I_n \text{ or } v \in I_n }}}^{\vphantom{\xi_k}} f(u,v).
\]

In the scenario of Theorem~\ref{th:unbias} for $m \geq 2$ tours and assuming priors $\mu(G)| \sigma^2 \sim \text{Normal}(\mu_0,\sigma^2/m_0)$, $\sigma^2 \sim \text{Inverse-gamma}(\nu_0/2,\nu_0\sigma_0^2/2)$ ($\sigma^2$ is the variance of $\hat{F}_1$), then the marginal posterior density of $\mu(G)$ as $m \to \infty$ converges {\em in distribution} to a non-standardized $t$-distribution
\begin{equation}\label{e:post}
\phi(x | \nu,\widetilde{\mu},\widetilde{\sigma}) = \frac{\Gamma\left(\frac{\nu+1}{2}\right)}{\Gamma\left(\frac{\nu}{2}\right) \, \widetilde{\sigma} \sqrt{\pi{}\nu}} \left(1+\frac{\left(x-\widetilde{\mu} \right)^2}{\widetilde{\sigma}^ 2{}\nu}\right)^{-\frac{\nu+1}{2}}
\end{equation}
with  degrees of freedom parameter $$\nu = \nu_0 + \lfloor \sqrt{m}\rfloor,$$location parameter
$$\widetilde{\mu} = \frac{m_0 \mu_0 + \lfloor \sqrt{m} \rfloor \hat{\mu}(\cD_m(\cS_n))}{m_0 + \lfloor \sqrt{m} \rfloor},$$
and scale parameter
$$\widetilde{\sigma} = \sqrt{ \frac{\nu_0 \sigma_0^2 + \sum_{k=1}^{\lfloor \sqrt{m} \rfloor} (\hat{F}_k - \hat{\mu}(\cD_m(\cS_n)))^2 + \frac{m_0 \lfloor \sqrt{m} \rfloor (\hat{\mu}(\cD_m(\cS_n)) - \mu_0)^2}{m_0+\lfloor \sqrt{m} \rfloor}}{(\nu_0 + \lfloor \sqrt{m} \rfloor)(m_0 +\lfloor \sqrt{m} \rfloor)}}.$$

\begin{remark}
Note that approximation in \eqref{e:post} is a Bayesian approach and Theorem \ref{th:unbias} is the frequentist counterpart. In fact, the motivation to form the Bayesian approach comes from the frequentist estimator ($\hat{F}_h$ samples). From the approximate posterior, the Bayesian MAP estimator for sufficiently large values of $m$ is
\[\hat{\mu}_{\text{MAP}}=\argmax_{x} \phi(x|v,\widetilde{\mu},\widetilde{\sigma})=\widetilde{\mu}.\]
Thus when $m_0=0$, the Bayesian estimator $\hat{\mu}_{\text{MAP}}$ is essentially the first term in the frequentist estimator $\hat{\mu}(\cD_m(\cS_n))$ (second term is calculated a priori), and hence both the estimators are same. In this paper we make use of the posterior distribution from the Bayesian approach to get the degrees of belief along with the common estimator from both the approaches.
\end{remark}

The above remark shows that the approximate posterior in \eqref{e:post} provides a way to access the confidence in the estimate $\hat{\mu}(\cD_m(\cS_n))$.
The Normal prior for the average gives the largest variance given a given mean. 
The inverse-gamma is a non-informative conjugate prior if the variance of the estimator is not too small~\cite{Gelman2006}, which is generally the case in our application. Other choices of prior, such as uniform, are also possible yielding different posteriors without closed-form solutions~\cite{Gelman2006}.
The posterior is conservative as $\hat{\mu}(\cD_m(\cS_n))$ is calculated from $m$ tours while the
posterior considers only $\sqrt{m}$ tours.
Being conservative, however, is advised as the posterior is for large values of $m$ and the conservative estimate better protects us from finite-sample anomalies and perform very well in practice as we see in Section~\ref{s:results}.

In what follows we provide the proofs of our main results.

\subsection{Proof of Theorem~\ref{th:unbias}}
\label{s:proof1}

The outline of this proof is as follows.
In Lemma~\ref{th:avg} we show that the estimate of $\mu(G)$ from each tour is unbiased.

\begin{lemma}\label{th:avg}
Let $X^{(k)}_1,\ldots,X^{(k)}_{\xi_k}$ be the nodes traversed by the $k$-th random walk tour on $G'$, $k \geq 1$ starting at super-node $\cS_n$. Then the following holds, $\forall k$,
\begin{equation}
E \Big[\sum_{t=2}^{\xi_k}f(X^{(k)}_{t-1},X^{(k)}_{t}) \Big]=\frac{2}{d_{\cS_n}} {\mu}(G').
\end{equation}
\label{th:edge_fn_estn}
\end{lemma}
\begin{proof}
The random walk starts from the super-node $\cS_n$, thus
\begin{multline}
E\Big[\sum_{t=2}^{\xi_k}f(X^{(k)}_{t-1},X^{(k)}_t) \Big] =\\  \sum_{(u,v) \in E'} E \Big[\Big(\text{No.\ of times Markov chain crosses $(u,v)$ in the tour}\Big) f(u,v)\Big]
\label{eq:proof_tour_edge}
\end{multline}
Consider a renewal reward process with inter-renewal time distributed as $\xi_k, k \geq 1$ and reward as the number of times Markov chain crosses $(u,v)$. From renewal reward theorem,
\begin{multline*}
\{ \text{Asymptotic frequency of transitions from $u$ to $v$}\}  \\
= \frac{E \Big[\Big(\text{No.\ of times Markov chain crosses $(u,v)$ in the tour}\Big) f(u,v)\Big]}{E[\xi_k]}
\end{multline*}
Here the left-hand side is essentially $2\pi_u p_{uv}$. Now \eqref{eq:proof_tour_edge} becomes
\begin{align*}
E \Big[\sum_{t=2}^{\xi_k}f(X^{(k)}_{t-1},X^{(k)}_t) \Big] &= \sum_{(u,v) \in E'} f(u,v)\, 2\pi_u \, p_{uv}\, E[\xi_k] \\
& = 2\sum_{(u,v) \in E'} f(u,v) \frac{d_u}{\sum_j d_j}\, \frac{1}{d_u}\, \frac{\sum_j d_j}{d_\csn}=\frac{2}{d_\csn} \sum_{(u,v) \in E'} f(u,v),
\end{align*}
which concludes our proof. 
\end{proof}
In what follows we prove Theorem~\ref{th:unbias} using Lemma~\ref{th:avg}.
\begin{proof}[Theorem~\ref{th:unbias}]
By Lemma \ref{th:avg} the estimator $W_k = \displaystyle \sum_{t=2}^{\xi_k-1}f(X^{(k)}_{t-1},X^{(k)}_t)$ is an unbiased estimate of $\mu(G')$. By the linearity of expectation the average estimator $\bar{W}(m)=m^{-1} \sum_{k=1}^{m} W_k$ is also unbiased.
Finally for the estimator
$$\hat{\mu}(\cD_m(\cS_n)) = \frac{d_{\cS_n}}{2m} \bar{W}(m) + \sum_{(u,v) \in E \text{ s.t.\ } u,v \in I_n}  f(u,v)$$
has average
\[
E[\hat{\mu}(\cD_m(\cS_n))] = \sum_{(u,v) \in E \text{ s.t.\ } u\not\in I_n \text{ or }v \not\in I_n}  f(u,v) + \sum_{(u,v) \in E \text{ s.t.\ } u,v \in I_n}  f(u,v) = \mu(G).
\]
Furthermore, by strong law of large numbers with $E[W_k]< \infty$, $\hat{\mu}(\cD_m(\cS_n)) \to \mu(G)$ a.s.\ for $m\to\infty$. This completes our proof.
\end{proof}

%% file: theory_estimator.tex

\subsection{Derivation of the approximate posterior}
\label{s:BayesianInference}
The derivation of \eqref{e:post} relies first on showing that $\hat{\mu}(\cD_m(\cS_n))$ has finite first and second moments.
We go further and in Lemma~\ref{th:variance_tour_estr} we introduce upper and lower bounds on the variance of $\hat{\mu}(\cD_m(\cS_n))$.
By Theorem~\ref{th:unbias} the first moment of $\hat{\mu}(\cD_m(\cS_n))$ is finite as $E[\hat{\mu}(\cD_m(\cS_n))] = \mu(G)$.
To show that the second moment is finite we prove that the estimate $W_k = \displaystyle \sum_{t=2}^{\xi_k-1}f(X^{(k)}_{t-1},X^{(k)}_t)$, $k \geq 1$,  whose variance is $m \cdot \text{Var}(\hat{\mu}(\cD_m(\cS_n)))$, has finite second moment.
The results in Lemma~\ref{th:variance_tour_estr} are of interest on their own because they establish a connection between
the estimator variance and the spectral gap.

\subsubsection{Impact of spectral gap on variance}
\label{s:var}
Let $\mathbf{S}=\bD^{1/2}\bP \bD^{-1/2}$, where $\bP$ is the random walk transition probability matrix as defined in Section~\ref{s:pre} and $\bD=\text{diag}(d_1,d_2,\ldots,d_{|V'|})$ is a diagonal matrix with the node degrees of $G'$. The eigenvalues $\{\lambda_i \}$ of $\bP$ and $\mathbf{S}$ are same and $1=\lambda_1> \lambda_2 \ge \ldots \ge \lambda_{|V'|} \geq -1$. Let $j$th eigenvector of $\mathbf{S}$ be $(w_{ji}),1\leq i \leq |V|$. Let $\delta$ be the spectral gap, $\delta:=1-\lambda_2$. Let the left and right eigenvectors of $\bP$ be $v_j$ and $u_j$ respectively. $d_{tot}:=\sum_{v \in V'} d_v$. Define $\langle f,g \rangle_{\hat{\pi}}=\sum_{(u,v) \in E'} \hat{\pi}_{uv} f(u,v) g(u,v)$, with $\hat{\pi}_{uv}=\pi_u p_{uv}$, and matrix $\bP^*$ with $(j,i)$th element as $p^*_{ji}=p_{ji} f(j,i)$. Also let $\hat{f}$ be the vector with $\hat{f}(j)=\sum_{i \in V'} p^*_{ji}$.
\begin{lemma}
\label{th:variance_tour_estr}
The following holds
\begin{enumerate}
\item[(i).] Assuming the function $f$ is bounded, $\displaystyle \max_{(i,j) \in E'} f(i,j) \leq B<\infty$, $B>0$ and for tour $k \geq 1$,
\begin{align*}
\allowdisplaybreaks[4]
\lefteqn{\text{\normalfont var}\left[ \sum_{t=2}^{\xi_k}f(X^{(k)}_{t-1},X^{(k)}_t) \right] } \nonumber \\
& \le  \frac{1}{d_\csn^2} \left(2d_{\text{tot}}^2B^2 \sum_{i \ge 2} \frac{w_{\csn i}^2}{ ( 1-\lambda_{i})}  -4 \mu^2(G_\csn) \right)-\frac{1}{d_\csn}B^2d_{\text{tot}}+B^2 \\
& <  B^2 \left(\frac{2 d_{\text{tot}}^2}{d_\csn^2 \delta}+1 \right).
\end{align*}
Moreover,
\[E\left[ \left(\sum_{t=2}^{\xi_k}f(X^{(k)}_{t-1},X^{(k)}_t) \right)^l \,\right] <\infty\quad \forall l \geq 0.\]
\item[(ii).]
\begin{align}
\allowdisplaybreaks[4]
\lefteqn{\text{\normalfont var}_{\cS_n} \left[ \sum_{t=2}^{\xi_k}f(X^{(k)}_{t-1},X^{(k)}_t)) \right] } \nonumber \\
& \geq   2\frac{d_{tot}}{d_\csn} \sum_{i=2}^r \frac{\lambda_i}{1-\lambda_i} \langle f, v_i \rangle_{\hat{\pi}} \,(u_i^\intercal \hat{f})+ \frac{1}{d_{\csn}} \sum_{(u,v) \in E'} f(u,v)^2 + \frac{1}{d_{tot}d_\csn} \left(\sum_{(u,v) \in E'} f(u,v)^2\right)^2 \nonumber \\
&  \qquad +\frac{1}{d_{tot} d_\csn} \sum_{u \in V'} d_u \Big(\sum_{u \sim v} f(u,v)\Big)^2-\frac{4}{d_\csn^2}\left(\sum_{(u,v) \in E'} f(u,v) \right)^2 \nonumber \\
&  \qquad -\frac{8}{d_{tot}} \left(\sum_{(u,v) \in E'} f(u,v) \right)^2 \sum_{i \ge 2} \frac{w_{\csn i}^2}{ ( 1-\lambda_{i})}-\frac{4}{d_{tot} d_\csn} \left(\sum_{(u,v) \in E'} f(u,v) \right)^2.
\label{eq:th_var_tour_estr_ii}
\end{align}
\end{enumerate}
\end{lemma}
\begin{proof}
(i). The variance of the estimator at tour $k \geq 1$ starting from node $\cS_n$ is
\begin{eqnarray}
{\text{var}_{\cS_n} \left[\sum_{t=2}^{\xi_k}f(X^{(k)}_{t-1},X^{(k)}_t)) \right] }\le  B^2 E [({\xi_k-1})^2]-\left( E \left[\sum_{t=2}^{\xi_k}f(X^{(k)}_{t-1},X^{(k)}_t)) \right]  \right)^2.
\label{eq:proof:var_i_1}
\end{eqnarray}
It is known from \cite[Chapter 2 and 3]{Aldous_2014} that
\begin{eqnarray*}
E [\xi_k^2] & = &\frac{2 {\sum_{i \ge 2} w_{\cS_n i}^2 ( 1-\lambda_{i})^{-1}}+1}{{\pi_{\cS_n}^2}}.
\end{eqnarray*}
Using Theorem~\ref{th:edge_fn_estn} eq.~\eqref{eq:proof:var_i_1} can be written as
\begin{eqnarray*}
\lefteqn{{\text{var} \left[\sum_{t=2}^{\xi_k}f(X^{(k)}_{t-1},X^{(k)}_t)) \right] }} \nonumber \\
&\leq &\frac{1}{d_{\cS_n }^2} \left(2d_{\text{tot}}^2B^2 ( \sum_{i \ge 2} w_{\cS_n  m}^2 ( 1-\lambda_{i})^{-1} ) -4 \mu^2(G') \right)-\frac{1}{d_{\cS_n }}B^2d_{\text{tot}}+B^2.
\end{eqnarray*}
The latter can be upper-bounded by $B^2(2 d_{\text{tot}}^2/(d_i^2 \delta)+1$).

For the second part, we have
\begin{align*}
E \left[ \left(\sum_{t=2}^{\xi_k}f(X^{(k)}_{t-1},X^{(k)}_t) \right)^l \,\right]
& \leq B^l E[({\xi_k-1})^l)] \leq C (E[(\xi_k)^l]+1),
\end{align*}
for a constant $C>0$ using $c_r$ inequality. From \cite{meyn2012markov}, it is known that there exists an $a>0$, such that $E[\exp(a\, \xi_k)]<\infty$, and this implies that $E[(\xi_k)^l]<\infty$ for all $l \geq 0$. This proves the theorem.

(ii). 
We denote $E_\pi f$ for $E_\pi[ f(Y_1,Y_2)]$ and $\text{Normal}(a,b)$ indicates Gaussian distribution with mean $a$ and variance $b$. With the trivial extension of the central limit theorem of Markov chains \cite{Kipnis1986} of node functions to edge functions, we have for the ergodic estimator $\bar{f}_n=n^{-1} \sum_{t=2}^{n} f(Y_{t-1},Y_t)$,
\begin{equation}
\sqrt{n} (\bar{f}_n-E_\pi f) \xrightarrow{d} \text{Normal}(0,\sigma^2_a), 
\label{eq:proof:var_ii_1}
\end{equation}
where
\begin{equation}
\sigma^2_{a}= \text{Var}(f(Y_{1},Y_2))+2 \sum_{l=2}^{n-1} \frac{(n-1)-l}{n}\, \text{Cov}(f(Y_{0},Y_1),f(Y_{l-1},Y_l)) <\infty \nonumber
\end{equation}
We derive $\sigma^2_a$ in Lemma \ref{lem:variance_asym}. Note that $\sigma^2_a$ is also the asymptotic variance of the ergodic estimator of edge functions.

Consider a renewal reward process at its $k$-th renewal, $k \geq 1$, with inter-renewal time $\xi_k$ and reward $W_k = \displaystyle \sum_{t=2}^{\xi_k}f(X^{(k)}_{t-1},X^{(k)}_t)$. Let $\bar{W}(n)$  be the average cumulative reward gained up to $m$-th renewal, i.e., $\bar{W}(m)=m^{-1} \sum_{k=1}^{m} W_k$. From the central limit theorem for the renewal reward process \cite[Theorem 2.2.5]{Tijms2003} after $n$ total number of steps, with $l_n = \arg\!\max_k \sum_{j=1}^k {\bf 1}(\xi_j \leq n)$, yields
\begin{equation}
\sqrt{n}(\bar{W}(l_n)-E_\pi f) \xrightarrow{d} \text{Normal} (0,\sigma_b^2),
\label{eq:proof:var_ii_2}
\end{equation}
with $\displaystyle \sigma_b^2=\frac{\nu^2}{E[\xi_k]} $ and
\begin{eqnarray}
\nu^2 &= & E[(W_k-\xi_k E_{\pi}f)^2]= E_i\left[\left(W_k-\xi_k\,\frac{E[W_k]}{E[\xi_k]} \right)^2\right] \nonumber \\
&=&\text{var}_{\cS_n}(W_k)+(E[W_k])^2+\left(\frac{E[W_k]}{E[\xi_k]} \right)^2 E[(\xi_k)^2]-2\frac{E[W_k]}{E[\xi_k]}E[W_k \xi_k]. \nonumber
\end{eqnarray}
In fact it can be shown that (see \cite[Proof of Theorem 17.2.2]{meyn2012markov})
\[|\sqrt{n} (\bar{f}_n-E_\pi f)-\sqrt{n}(\bar{W}(l_n)-E_\pi f)| \to 0 \quad \text{ a.s.} \, . \]
Therefore $\sigma^2_a=\sigma^2_b$. Combing this result with Lemma \ref{lem:variance_asym} shown in the appendix we get \eqref{eq:th_var_tour_estr_ii}. 
\end{proof}
We are now ready to derive the approximation \eqref{e:post}.
%
\begin{proof}
Let $m' = \lfloor \sqrt{m} \rfloor$.
Given
\[
\hat{F}_h = \frac{d_{\cS_n}}{2 \lfloor\sqrt{m}\rfloor} \sum_{k=((h-1)\lfloor \sqrt{m}\rfloor + 1)}^{h \lfloor\sqrt{m}\rfloor} \sum_{t=2}^{\xi_h}  f(X^{(k)}_{t-1},X^{(k)}_{t})  \, .
\]
and $\{\hat{F}_h\}_{h=1}^{m'}$ and because the tours are i.i.d.\,
$\hat{\mu}(\cD_{\lfloor \sqrt{m} \rfloor}(S_n))$
the marginal posterior density of $\mu$ is
\[
P[\mu |\{\hat{F}_h\}_{h=1}^{m'}] = \int_0^\infty P[\mu|\sigma^2,\{\hat{F}_h\}_{h=1}^{m'}] P[\sigma^2 | \{\hat{F}_h\}_{h=1}^{m'}]  d\sigma^2 \, .
\]
For now assume that $\{\hat{F}_h\}_{h=1}^{m'}$ are i.i.d.\ normally distributed random variables, and let $$\hat{\sigma}_{m'} = \sum_{h=1}^{m'} (\hat{F}_h - \hat{\mu}(\cD_{m'}(S_n)))^2,$$ then~\cite[Proposition C.4]{Jackman2009}
\begin{align*}
\mu | \sigma^2,\{\hat{F}_h\}_{h=1}^{m'} & \sim \text{Normal}\left( \frac{m_0 \mu_0 +  \sum_{h=1}^{m'}\hat{F}_h }{m0+m'} , \frac{\sigma^2}{m_0 + m'} \right) \, , \\
 \sigma^2 | \{\hat{F}_h\}_{h=1}^{m'} & \sim \text{Inverse-Gamma}\left( \frac{\nu_0 + m'}{2} ,  \frac{\nu_0 \sigma_0^2 + \hat{\sigma}_{m'} + \frac{m_0 m'}{m_0 + m'} (\mu_0 - \hat{\mu}(\cD_m(S_n)))^2 }{2} \right)
\end{align*}
are the posteriors of parameters $\mu$ and $\sigma^2$, respectively.
The non-standardized $t$-distribution can be seen as a mixture of normal distributions with equal mean and random variance inverse-gamma distributed~\cite[Proposition C.6]{Jackman2009}.
Thus, if $\{\hat{F}_h\}_{h=1}^{m'}$ are i.i.d.\  normally distributed then the posterior of $\hat{\mu}(\cD_{\lfloor \sqrt{m} \rfloor}(S_n))$ is a
non-standardized $t$-distributed with parameters
\begin{align}
t\Bigg(&\mu = \frac{m_0 \mu_0 +  \sum_{h=1}^{m'}\hat{F}_h }{m0+m'}, \sigma^2 = \frac{\nu_0 \sigma_0^2 + \sum_{k=1}^{\lfloor \sqrt{m} \rfloor} (\hat{F}_k - \hat{\mu}(\cD_m(\cS_n)))^2 + \frac{m_0 \lfloor \sqrt{m} \rfloor (\hat{\mu}(\cD_m(\cS_n)) - \mu_0)^2}{m_0+\lfloor \sqrt{m} \rfloor}}{(\nu_0 + \lfloor \sqrt{m} \rfloor)(m_0 +\lfloor \sqrt{m} \rfloor)},\nonumber \\
&\nu =  \nu_0 + \lfloor \sqrt{m} \rfloor \Bigg).
\end{align}
Left to show is that $\{\hat{F}_h\}_{h=1}^{m'}$ are converge  {\em in distribution} to i.i.d.\  normal random variables as $m \to \infty$.
As the spectral gap of $G_{\cS_n}$ is greater than zero, $|\lambda_1 - \lambda_2| > 0$, Lemma~\ref{th:variance_tour_estr} showns that for
$W_k = \displaystyle \sum_{t=2}^{\xi_k-1}f(X^{(k)}_{t-1},X^{(k)}_t)$ then
 \[ \sigma_W^2 = \text{\normalfont Var} ( W_k ) < \infty \, ,  \quad \forall k \, .\]
From the renewal theorem we know that $\{W_k\}_{k=1}^m$ are i.i.d.\ random variables and thus any subset of these variables is also i.i.d..
By construction $\hat{F}_1, \ldots, \hat{F}_{m'}$ are also i.i.d.\ with mean $\mu(G_{\cS_n})$ and finite variance $0 < \sigma_{m'}^2 < \infty$.
Applying the Lindeberg-L\'evy central limit theorem~\cite[Section~17.4]{cramer1999mathematical} yields
\[
 \sqrt{m'} (\hat{F}_h - \mu(G'))  \stackrel{d}{\to} \text{Normal}(0,\sigma_W^2), \quad \forall h \, ,
\]
where $\text{var}(\hat{F}_h) = \sigma_{m'}^2$.
Thus, in the limit as $m \to \infty$ and $m' \to \infty$ (recall that $m' = \lfloor \sqrt{m} \rfloor $),
the variables $\{\hat{F}_h\}_{h=1}^{m'}$ are i.i.d.\ normally distributed with
\[
\hat{F}_h \sim \text{Normal}(\mu(G'),\sigma_W^2/m') \,, \quad \forall h \: ,
\]
and
$$
\sum_{\mathclap{\substack{(u,v) \in E \text{ s.t.\ } \\ u\in I_n \text{ or } v \in I_n }}}^{\vphantom{\xi_k}} f(u,v)
$$
is constant and known, which concludes our proof.

\end{proof}

%% file: illustrative_examples.tex
\section{Illustrative example}
\label{s:illustrative_examples}
Here we consider the classical example of low-conductance graph: the {\em dumbbell} graph.
Here we Illustrate how the super-node solves the {\em variance} problem for random walk tours on graphs. 
A dumbbell graph consists of two complete graphs $K_{n}$ on $n$ vertices connected by a single edge. 
The spectral gap $\delta = (1-\lambda_2)$, where $\lambda_2$ is the second largest absolute eigenvalue of $\bP$, is roughly $\delta = O(1/n^2)$.

It is known that the variance of the return time $\xi_k$ of tour $k > 0$ is related to $\delta$ as \cite{Aldous_2014}
\[\text{Var}(\xi_k) \leq \frac{2}{\delta\, \pi_{\cS_n}^2}+\frac{1}{\pi_{\cS_n}}.\]
Drawing nodes from both components to create the super-node, the new graph $G'$ with the super-node will be more connected and hence $\delta$ improves, and so does the variance.

Another way to view the impact of forming the super-node is that of the cover time of a random walk on dumbell graph.
Without the super-node the cover time is $\Theta(n^2)$. 
If $k=O(\log n)$ random walks run in parallel with some conditions on distributing them, the covering time can be reduced to $O(n)$ \cite{Alon2011}. 
In this view the super-node tours acts as multiple parallel random walks that quickly cover more of graph with less effort.

%% file: real_examples.tex
\section{Experiments on Real-world Networks}
\label{s:results}
In this section we demonstrate the effectiveness of the theory developed above with the experiments on real data sets of various social networks \footnote{The developed software is available here: \url{http://www-sop.inria.fr/members/Jithin.Sreedharan/HypRW.zip}}. We assume the contribution from super-node to the true value is known a priori and hence we look for $\mu(G_{\cS_n})$ in the experiments. In the case that the edges of the super-node are unknown, the estimation problem is easier and can be taken care separately. One option is to start multiple random walks in the graph and form connected subgraphs. Later, in order to estimate the bias created by this subgraph, do some random walk tours from the largest degree node in each of these sub graph and use the idea in Theorem \ref{th:edge_fn_estn}.

In the figures we display both approximate posterior and empirical posterior generated from $\hat{F}_h$. 
For the approximate posterior, we have used the following parameters $m_0=0,\nu_0=0,\mu_0=0,\sigma_0=1$. The green line in the plots shows the actual value $\mu(G_{\cS_n})$.

In the numerical experiments, the super-node is formed in one of following ways: a) uniformly sample $k$ nodes from the network without replacement; b) run random walk crawl starting from any node and cover around $10\%$ of the graph, and form the super-node with the $k$ largest degree nodes. In both the ways, if the network is disconnected, super-node should be initially created with at least one node from each of the connected component.
\subsection{Friendster}
First we study a network of moderate size, a connected subgraph of Friendster network with $64,\!600$ nodes and $1,\!246,\!479$ edges  (data publicly available at the SNAP repository~\cite{Snapdata}). Friendster is an online social networking website where nodes are individuals and edges indicate friendship. Here, we consider two types of functions:
\begin{enumerate}
\item $f_1=d_{X_t}.d_{X_{t+1}}$
\item $f_2=\begin{cases}1 \text{ if } d_{X_t} + d_{X_{t+1}} > 50 \\
0 \text{ otherwise}
\end{cases}$
\end{enumerate}
These functions reflect assortative nature of the network.
The super-node is formed from $10,\!000$ uniformly sampled nodes just as a way to test our method. Figures \ref{fig:fn_1_2} and \ref{fig:fn_2_6} display the results for functions $f_1$ and $f_2$, respectively. A good match between the approximate and empirical posteriors can be observed from the figures. Moreover the true value $\mu(G_{\cS_n})$ is also fitting well with the plots. The percentage of graph crawled is $24.44 \% $ in terms of edges and this drops to $7.43 \%$ if we use random walk based super-node formation.

\begin{figure}
\hspace{0 cm}
\centering
\includegraphics[scale=0.4]{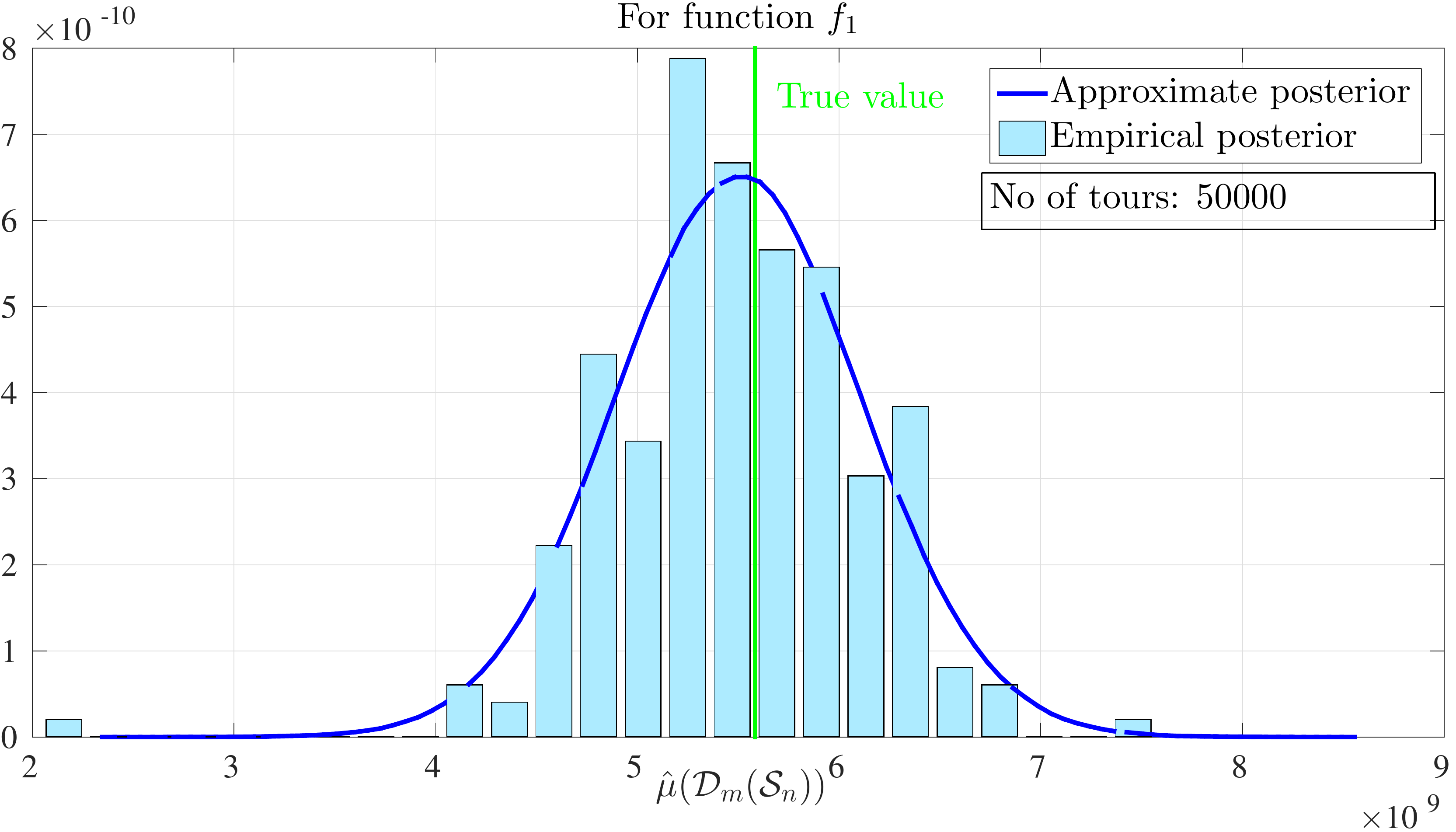}
\caption{Friendster subgraph, function $f_1$ }
\label{fig:fn_1_2}
\end{figure}

\begin{figure}
\hspace{0 cm}
\centering
\includegraphics[scale=0.4]{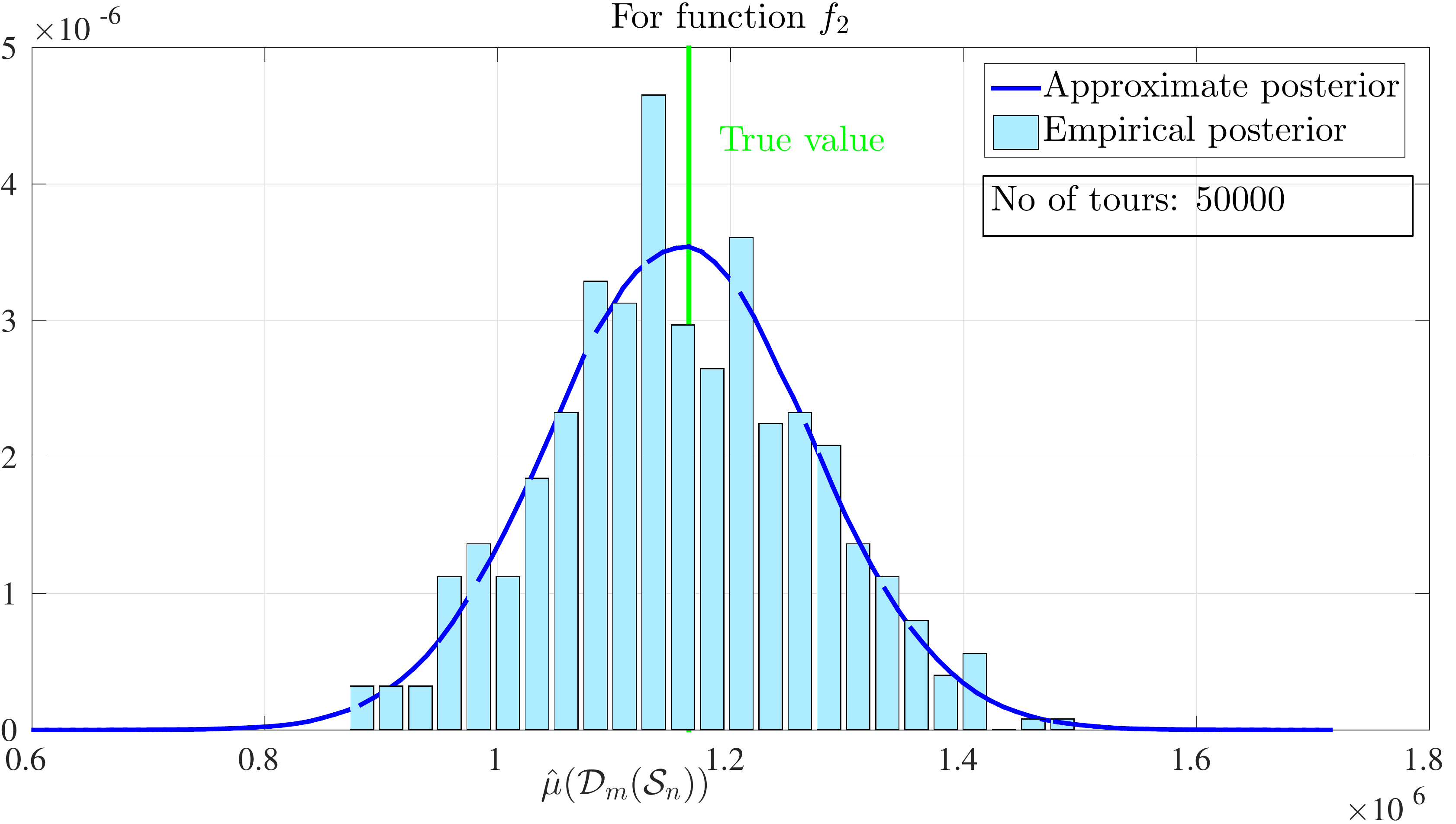}
\caption{Friendster subgraph, function $f_2$}
\label{fig:fn_2_6}
\end{figure}

\subsection{Dogster network}
The aim of this example is to check whether there is any affinity for making connections between the owners of same breed dogs \cite{dogsdata}. The network data is based on the social networking website Dogster. Each user (node) indicates the dog breed; the friendships between dogs form the edges. Number of nodes is $415,\!431$ and number of edges is $8,\!265,\!511$.

In Figure \ref{fig:dog_pals}, two cases are plotted. Function $f_1$ counts the number of connections with different breeds as pals and function $f_2$ counts connections between same breeds. The super-node is formed by $10,\!000$ nodes which are uniformly selected at random. The percentage of the graph crawled in terms of edges is $5.02\%$ and  in terms of nodes is $37.17\%$. While using the random walk based super-node formation, the graph crawled drops to $2.72 \%$  (in terms of edges) and $14.86 \%$ (in terms of nodes) with the same super-node size.  But these values can be reduced much further if we allow a bit less precision in the match between approximate distribution and histogram.
\begin{figure}
\hspace{0 cm}
\centering
\includegraphics[scale=0.4]{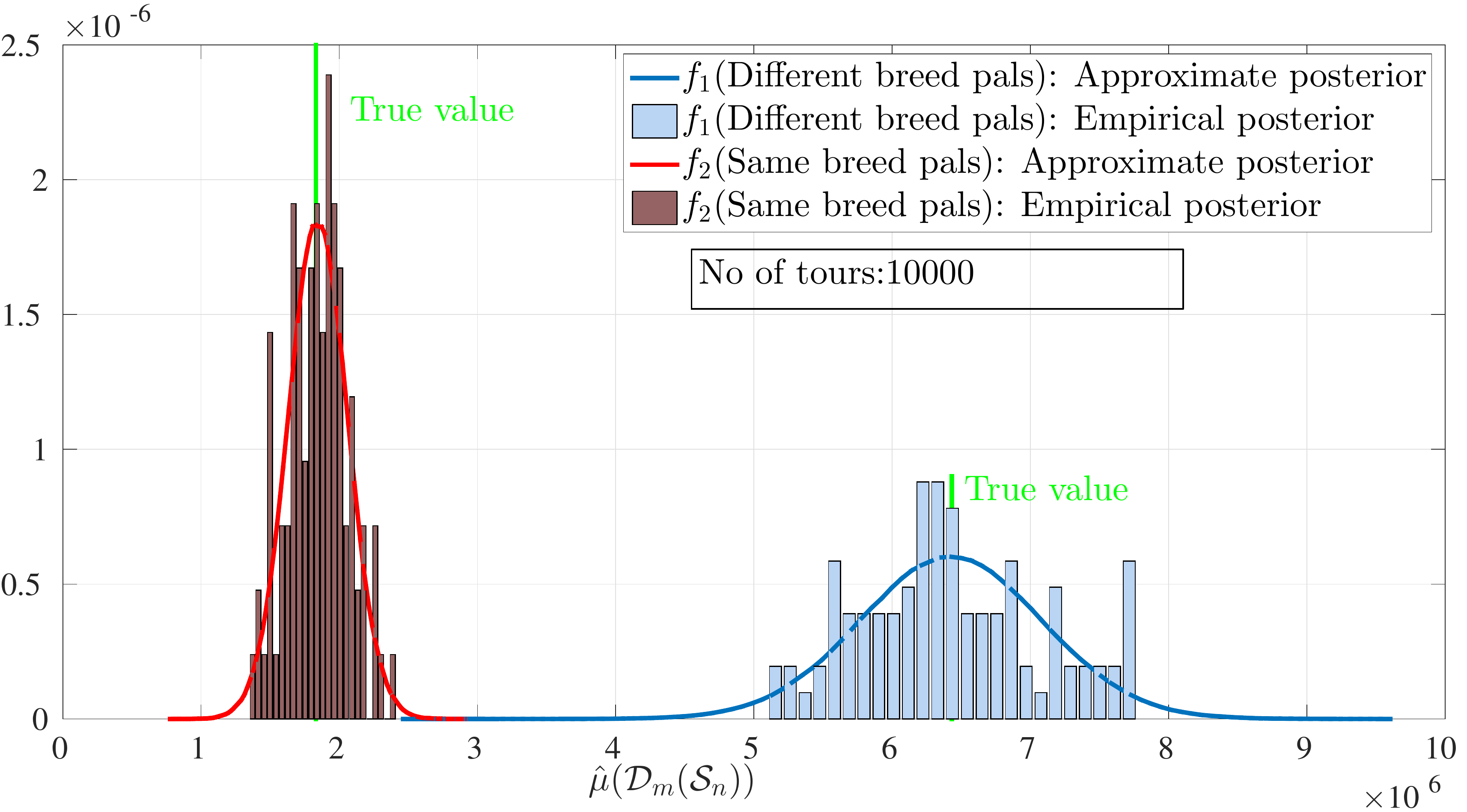}
\caption{Dog pals network}
\label{fig:dog_pals}
\end{figure}

In order to better understand the correlation in forming edges, we now consider the \textit{configuration model}. The configuration model is formed as follows: all the edges in the graph is cut and what is left is half edges for each node, and the number of half edges is the degree of the node. Now these half edges are paired uniformly. Such a configuration will create a graph whose edges are formed without any correlation between the end-nodes. We run our estimator on the configuration model and plot the histogram and distribution as we did for the original graph. Figure \ref{fig:dog_pals_config_f_2} compare function $f_2$ for the configuration model and original graph. The figure shows that in the correlated case (original graph), the affinity to form connection between same breed owners is around $7.5$ times more than that in the uncorrelated case. Figure \ref{fig:dog_pals_config_f_1} shows similar figure in case of $f_1$.
It is important to note that one can create a configuration network model from the crawl and the knowledge of the complete
network is not necessary. In the figures, we show the estimated true value from the configuration model created with the subgraph sampled by the estimator proposed in this paper (blue line in Figure \ref{fig:dog_pals_config_f_2} and red line in Figure \ref{fig:dog_pals_config_f_1}). The estimator is simply $\mu_C(\cD(\csn))=\sum_{(u,v) \in E_c} f(u,v) |E|/|E_c|,$ where $E_c$ is the edge set in the configuration model subgraph. The number of edges $|E|$ can be calculated from the techniques in \cite{Cooper2013}. Interestingly, this estimated value matches with the true value of the configuration model generated from the degree sequence of the entire graph.
\begin{figure}
\hspace{0 cm}
\centering
\includegraphics[scale=0.4]{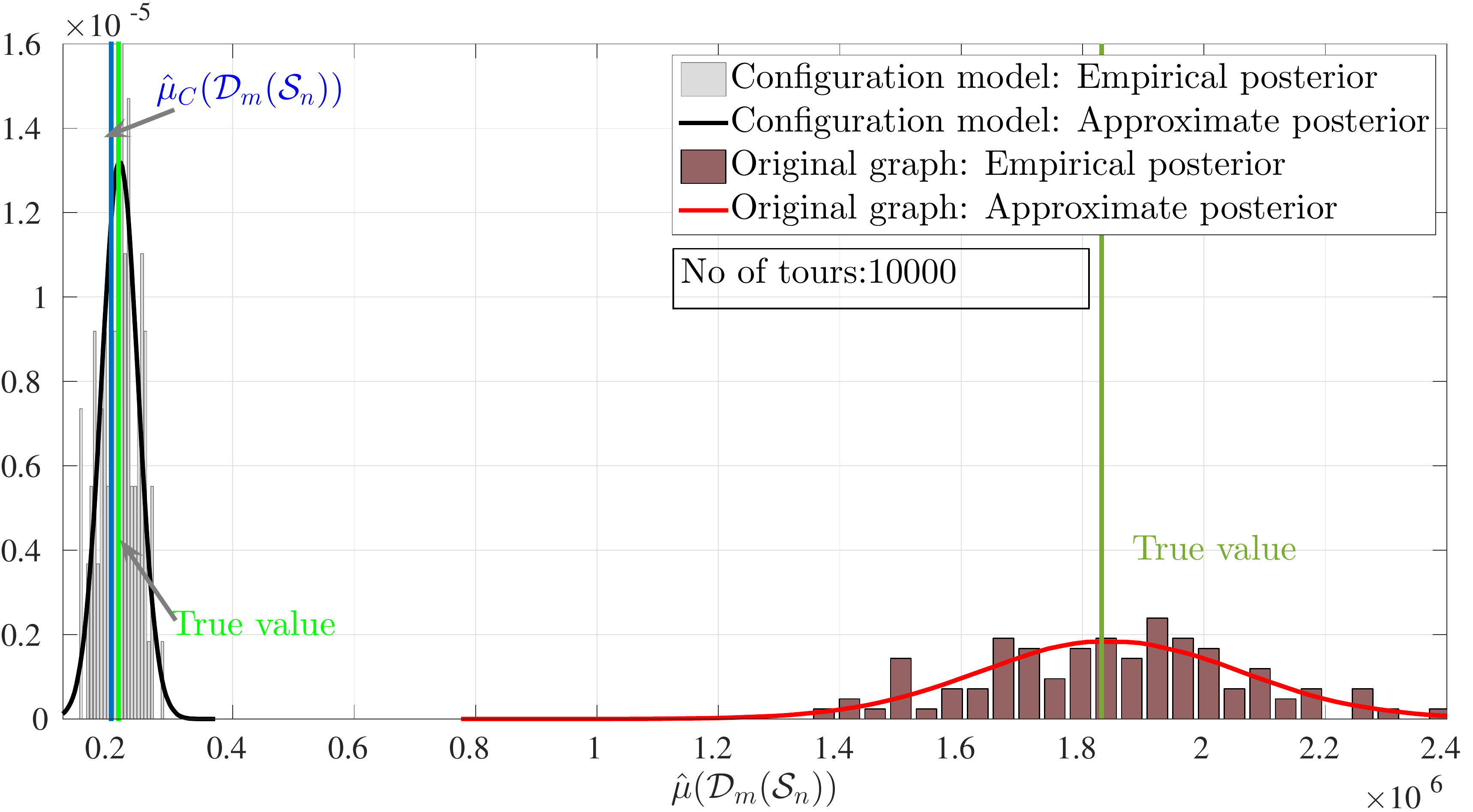}
\caption{Dog pals network: Comparison between configuration model and original graph for $f_2$, number of connection between same breeds}
\label{fig:dog_pals_config_f_2}
\end{figure}
\begin{figure}
\hspace{0 cm}
\centering
\includegraphics[scale=0.4]{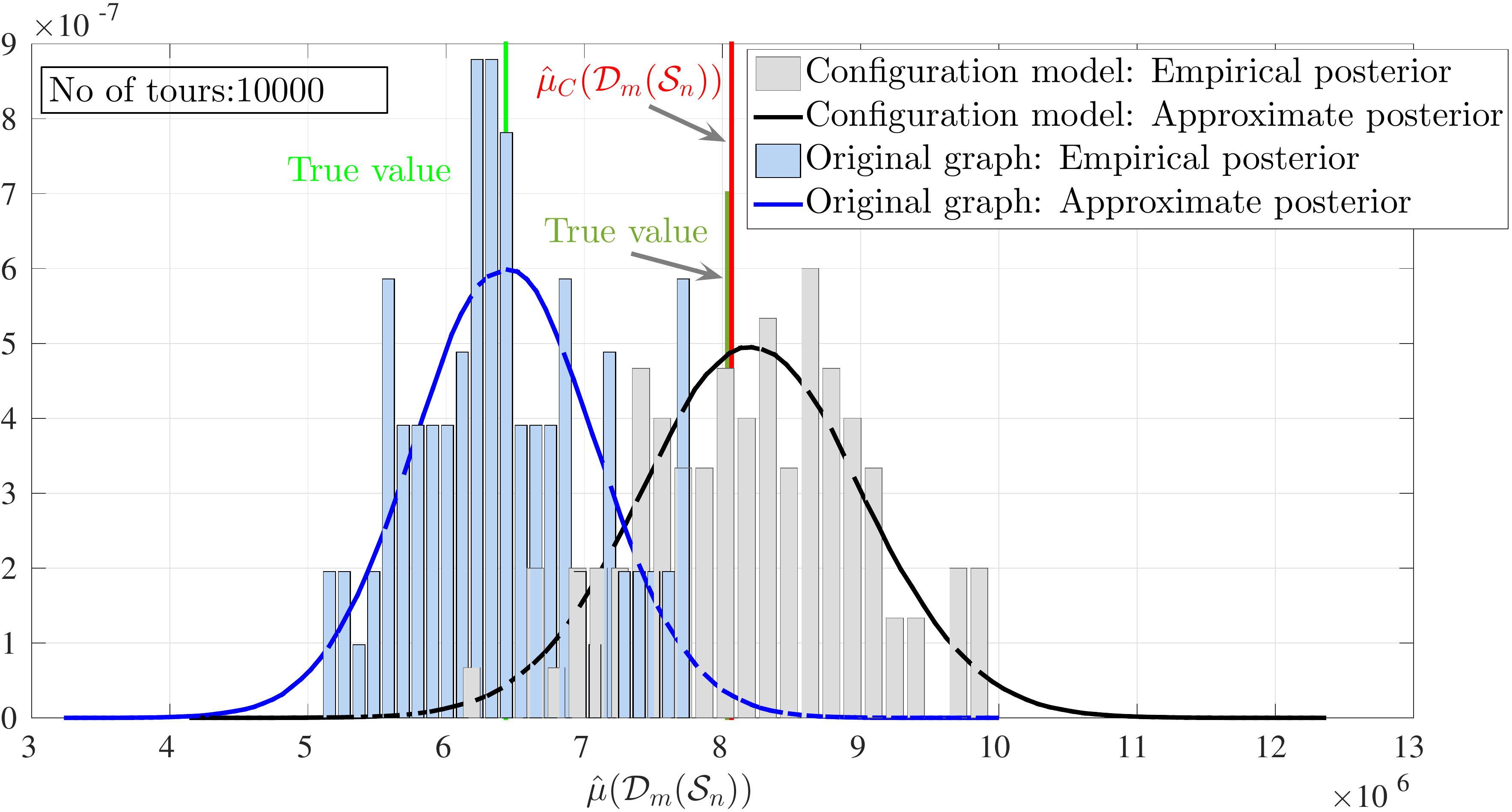}
\caption{Dog pals network: Comparison between configuration model and original graph for $f_1$, number of connection between different breeds}
\label{fig:dog_pals_config_f_1}
\end{figure}

\subsection{ADD Health data}
Though our main result in \eqref{e:post} is the approximation for large values of $m$, in this section we check with a small dataset. We consider ADD network project (\url{http://www.cpc.unc.edu/projects/addhealth}), a friendship network among high school students in US. The graph has $1545$ nodes and $4003$ edges.

We take two types of functions. Figure \ref{fig:ADD_gender} shows the affinity in same gender or different gender friendships and Figure \ref{fig:ADD_race} displays the inclination towards same race or different race in friendships. The random walk tours covered around $10 \%$ of the graph. We find that the theory works reasonably well for this network data. We have not added the empirical posterior in the figures since for such small sample sizes, the empirical distribution can not converge. 

\begin{figure}
\hspace{0 cm}
\centering
\includegraphics[scale=0.4]{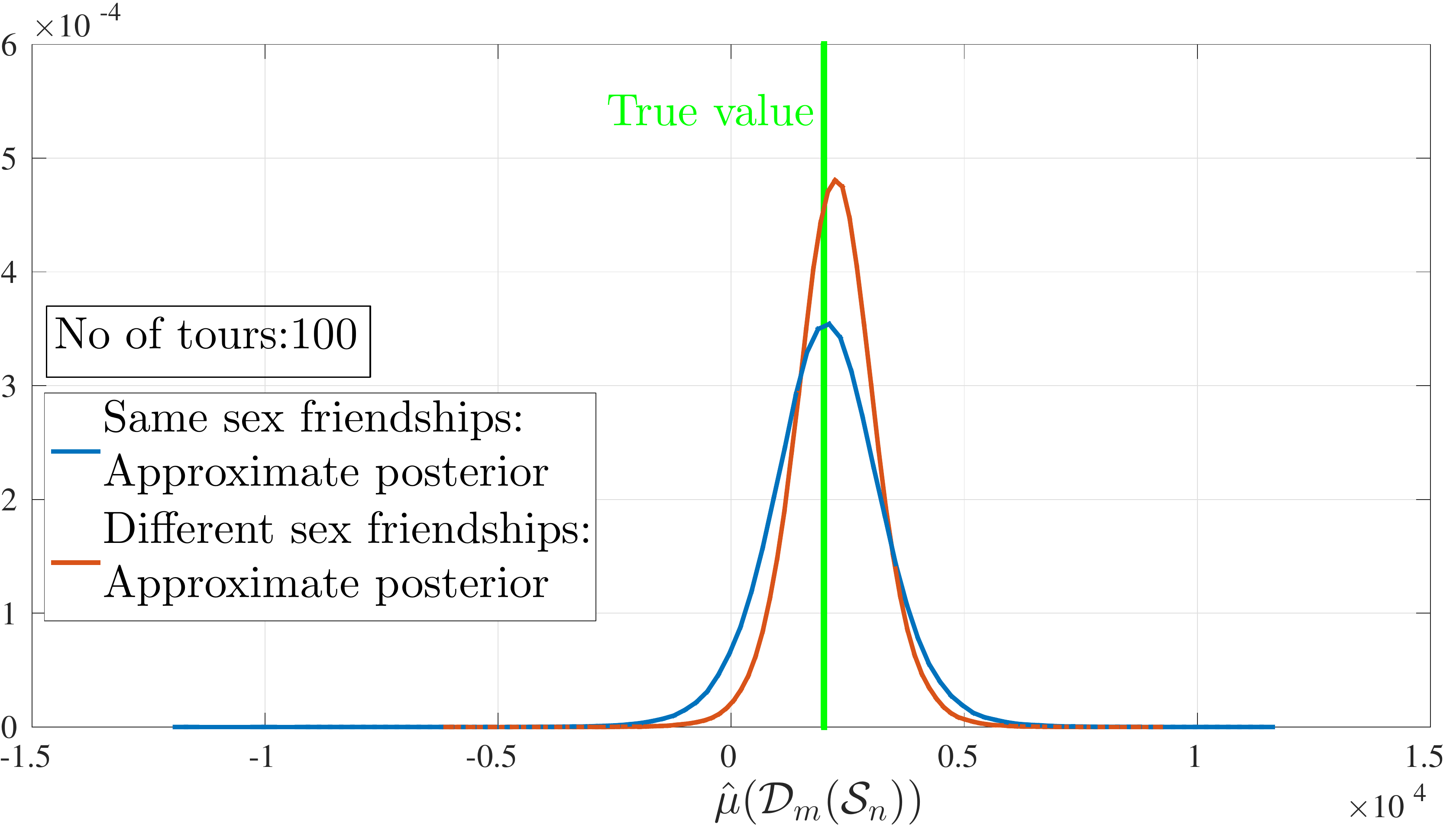}
\caption{ADD network: effect of gender in relationships}
\label{fig:ADD_gender}
\end{figure}

\begin{figure}
\hspace{0 cm}
\centering
\includegraphics[scale=0.4]{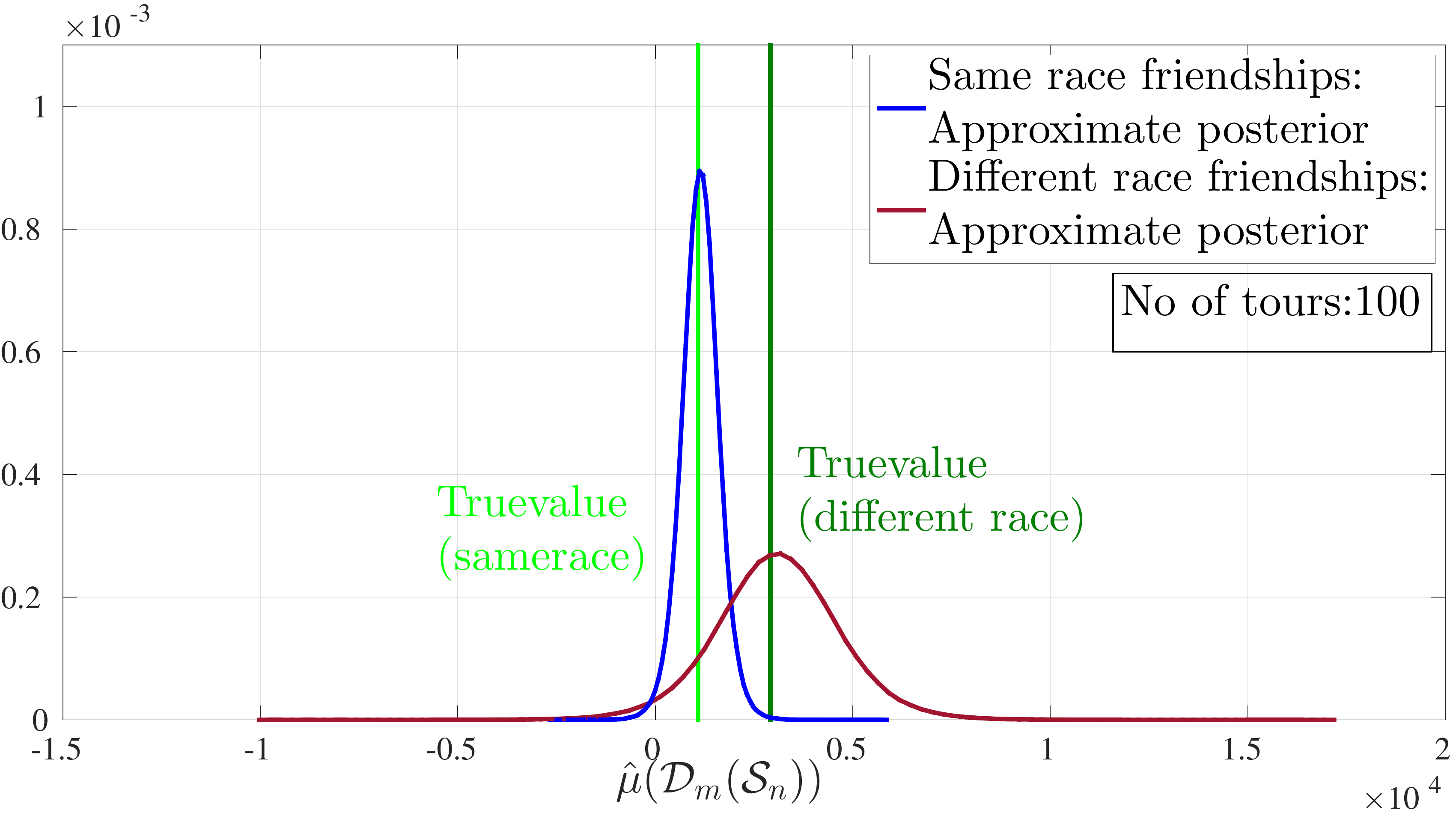}
\caption{ADD network: effect of race in friendships}
\label{fig:ADD_race}
\end{figure} 

%% file: conclusions.tex

\section{Conclusions}
\label{s:conclusions}
In this work we have studied an efficient way to make the statistical inference on networks. We introduced a technique to make a quicker inference by crawling a very small percentage of a network which may be disconnected or having low conductance. The proposed method can also be performed in parallel. In the paper, first we presented a non-asymptotic unbiased estimator and derived the bounds on its variance showing the connection with the spectral gap. Later we proved that the approximate posterior of the estimator given the crawled data converges to a non-centralized student's t distribution. The numerical experiments on real-world networks of different sizes, large and small, demonstrate the correctness of the estimator. In particular, the simulations clearly show that the derived posterior distribution fits very well with the data even while crawling a small percentage of the graph. 

%% file: appendix_tours.tex

\begin{lemma}
\label{lem:variance_asym}
\begin{eqnarray*}
\lefteqn{\lim_{n \to \infty} \frac{1}{n} \text{\normalfont var}_{\pi} \left(\sum_{k=2}^n f(Y_{k-1},Y_k) \right)}\nonumber \\
& = & 2 \sum_{i=2}^r \frac{\lambda_i}{1-\lambda_i} \langle f, v_i \rangle_{\hat{\pi}} \,(u_i^\intercal \hat{f})+ \frac{1}{d_{tot}} \sum_{(i,j) \in E} f(i,j)^2 + \frac{1}{d_{tot}^2} (\sum_{(i.j) \in E} f(i,j)^2)^2 \\
& &  \qquad +\frac{1}{d_{tot}^2} \sum_{i \in V} d_i \Big(\sum_{i \sim j} f(i,j)\Big)^2
\end{eqnarray*}
\end{lemma}
\begin{proof}
We extend the arguments in the proof of \cite[Theorem 6.5]{B99} to the edge functions.
When the initial distribution is $\pi$, we have
\begin{align}
\allowdisplaybreaks[4]
\lefteqn{\lim_{n \to \infty} \frac{1}{n} \text{var}_{\pi} \left(\sum_{k=2}^n f(Y_{k-1},Y_k) \right)} \nonumber \\
&= \frac{1}{n} \left(\sum_{k=2}^n \text{var}_{\pi} (f(Y_{k-1},Y_k))+2\sum_{\substack{k,j=2 \\ k<j}} \text{cov}_{\pi} (f(Y_{k-1},Y_k),f(Y_{j-1},Y_j)) \right) \nonumber \\
&= \text{var}_{\pi}(f(Y_{k-1},Y_k))+2 \sum_{l=2}^{n-1} \frac{(n-1)-l}{n}\, \text{cov}_{\pi}(f(Y_{0},Y_1),f(Y_{l-1},Y_l)). 
\label{eq:prf_lemma_var_1} 
\end{align}
Now the first term in \eqref{eq:prf_lemma_var_1} is
\begin{equation}
\text{var}_{\pi}(f(Y_{k-1},Y_k))= \langle f,f \rangle_{\hat{\pi}}-\langle f, \Pi \hat{f} \rangle_{\hat{\pi}},
\end{equation}
where $\Pi=\bone \pi^\intercal$.

For the second term in \eqref{eq:prf_lemma_var_1},
\begin{equation}
\text{cov}_{\pi}(f(Y_{0},Y_1),f(Y_{l-1},Y_l))=E_{\pi}(f(Y_{0},Y_1),f(Y_{l-1},Y_l))-(E_{\pi}[f(Y_{0},Y_1)])^2.
\end{equation}
\begin{align}
\allowdisplaybreaks[4]
\lefteqn{E_{\pi}(f(Y_{0},Y_1),f(Y_{l-1},Y_l))} \nonumber \\
&=  \sum_i \sum_j \sum_k \sum_m \pi_i\, p_{ij}\, p_{jk}^{(l-2)}\, p_{km} f(i,j) f(k,m) \nonumber \\
&=  \sum_i \sum_j \sum_k \pi_i \, p_{ij}\,f(i,j) \, p_{jk}^{(l-2)} \, \hat{f}(k) \nonumber \\
&=  \sum_{(i,j) \in E} \hat{\pi}_{ij}\,\,f(i,j)\, (P^{(l-2)}\hat{f})(j) \nonumber \\
&=  \langle f, P^{(l-2)} \hat{f} \rangle_{\hat{\pi}}.
\end{align}
Therefore,
\[\text{cov}_{\pi}(f(Y_{0},Y_1),f(Y_{l-1},Y_l))= \langle f, (\bP^{(l-2)}-\Pi) \hat{f}\rangle_{\hat{\pi}}. \]
Taking limits, we get
\begin{align}
\allowdisplaybreaks[4]
\lefteqn{\lim_{n \to \infty} \sum_{l=2}^{n-1} \frac{n-l-1}{n} \, (\bP^{(l-2)}-\Pi)} \nonumber \\
&= \lim_{n \to \infty} \sum_{k=1}^{n-3} \frac{n-k-3}{n}\, (\bP^k-\Pi)+(\bI-\Pi) \nonumber \\
&\stackrel{(a)}= \lim_{n \to \infty} \sum_{k=1}^{n-1} \frac{n-k}{n}\, (\bP^k-\Pi)+(\bI-\Pi)-\lim_{n \to \infty} \frac{3}{n}\sum_{k=1}^{n-3} (\bP^k-\Pi) \nonumber \\
&=  (\bZ-\bI)+(\bI-\Pi)= \bZ-\Pi,
\end{align}
where the first term in $(a)$ follows from the proof of \cite[Theorem 6.5]{B99} and since $\lim_{n \to \infty} (\bP^n-\Pi)=0$, the last term is zero using Cesaro's lemma \cite[Theorem 1.5 of Appendix]{B99}.

We have,
\[\bZ=\bI+\sum_{i=2}^r \frac{\lambda_i}{1-\lambda_i} v_i u_i^\intercal,\]
Thus
\begin{align}
\allowdisplaybreaks[4]
\lefteqn{\lim_{n \to \infty} \frac{1}{n} \text{\normalfont var}_{\pi} \left(\sum_{k=2}^n f(Y_{k-1},Y_k) \right)} \nonumber \\
& =  \langle f,f \rangle_{\hat{\pi}}-\langle f, \Pi \hat{f} \rangle_{\hat{\pi}}+2\langle  f,\left(\bI+\sum_{i=2}^r \frac{\lambda_i}{1-\lambda_i} v_i u_i^\intercal-\Pi  \right) \hat{f} \rangle_{\hat{\pi}} \nonumber \\
& =  \langle f,f \rangle_{\hat{\pi}}+\langle f, \Pi \hat{f} \rangle_{\hat{\pi}}+2 \langle  f,\hat{f} \rangle_{\hat{\pi}}+2 \sum_{i=2}^r \frac{\lambda_i}{1-\lambda_i} \langle f, v_i \rangle_{\hat{\pi}} \,(u_i^\intercal \hat{f}) \nonumber \\
& =  \frac{1}{d_{tot}} \sum_{(i,j) \in E} f(i,j)^2+ \frac{1}{d_{tot}^2} (\sum_{(i.j) \in E} f(i,j)^2)^2+\frac{1}{d_{tot}^2} \sum_{i \in V} d_i \Big(\sum_{i \sim j} f(i,j)\Big)^2 \nonumber \\
&   \qquad +2 \sum_{i=2}^r \frac{\lambda_i}{1-\lambda_i} \langle f, v_i \rangle_{\hat{\pi}} \,(u_i^\intercal \hat{f}) 
\end{align}

\end{proof}